\documentclass[10pt]{article}
\usepackage{fullpage}
\usepackage{amsmath,amsfonts,amssymb,amsthm}
\usepackage{graphicx} 
\usepackage{filecontents}
\usepackage[usenames]{color}

\newenvironment{enumerate*}%
  {\vspace{-2ex} \begin{enumerate} %
     \setlength{\itemsep}{-1ex} \setlength{\parsep}{0pt}}%
  {\end{enumerate}}
 
\newenvironment{itemize*}%
  {\vspace{-2ex} \begin{itemize} %
     \setlength{\itemsep}{-1ex} \setlength{\parsep}{0pt}}%
  {\end{itemize}}
 
\newenvironment{description*}%
  {\vspace{-2ex} \begin{description} %
     \setlength{\itemsep}{-1ex} \setlength{\parsep}{0pt}}%
  {\end{description}}

\let\Pr\relax
\DeclareMathOperator*{\Pr}{\mathbb{P}}

\newtheorem{theorem}{Theorem}
\newtheorem{lemma}{Lemma}

\newtheorem{invariant}{Invariant}

\usepackage{amsmath}
\usepackage{amssymb}
\usepackage{amsthm}
\usepackage{color}

\usepackage{enumerate}

\author{Zhengyu Wang\footnote{School of Engineering and Applied Sciences, Harvard University, zhengyuwang@g.harvard.edu. Supported by NSF Grant CCF-1350670.}}
\title{An Improved Randomized Data Structure for Dynamic Graph Connectivity}

\begin{document}

\maketitle

\abstract

We present a randomized algorithm for dynamic graph connectivity. With failure probability less than $1/n^c$ (for any constant $c$ we choose), our solution has worst case running time $O(\log^3 n)$ per edge insertion, $O(\log^4 n)$ per edge deletion, and $O(\log n/\log\log n)$ per query, where $n$ is the number of vertices. The previous best algorithm has worst case running time $O(\log^4 n)$ per edge insertion and $O(\log^5 n)$ per edge deletion. The improvement is made by reducing the randomness used in the previous result, so that we save a $\log n$ factor in update time. 

Specifically, \cite{kapron2013dynamic} uses $\log n$ copies of a data structure in order to boost a success probability from $1/2$ to $1-n^{-c}$. We show that, in fact though, because of the special structure of their algorithm, this boosting via repetition is unnecessary. Rather, we can still obtain the same correctness guarantee with high probability by arguing via a new invariant, without repetition.

\newpage

\section{Overview }\label{sec:overview}

\subsection{Dynamic Graph Connectivity Problem}

We are given an undirected graph $G=(V,E)$ with $|V|=n$ and initially $E=\emptyset$. We have $m=poly(n)$ operations coming online, each of which is either an update inserting/deleting an edge into/from $G$,  or a query asking whether two vertices are connected to each other. We assume the set of vertices cannot be updated. More formally, there are $m$ operations of the following three possible forms:

\begin{itemize}
\item Insert($e$): insert edge $e$ into the graph. It is guaranteed that $e$ is not in $G$.
\item Delete($e$): delete edge $e$ from the graph. It is guaranteed that $e$ is in $G$. 
\item Query($u, v$): return whether vertex $u$ and $v$ are connected in $G$.
\end{itemize}

The dynamic graph connectivity problem asks us to maintain a dynamic data structure to answer all the queries, at the same time minimizing processing time for each operation. The processing time can be measured in both the amortized case and worst case. 

For the amortized case, \cite{henzinger1996improved} proposes an algorithm with amortized time complexity $O(\log^2 n)$ per update. A further result by \cite{thorup2000near} gives a nearly optimal randomized solution with expected amortized time $O(\log n (\log \log n)^3)$ per update, and $O(\log n/\log \log \log n)$ per query. 

However, for the worst case, the current best deterministic algorithm achieves $O(\sqrt{n}/\log^{1/4} n)$ worst case update time and constant query time \cite{kejlberg2015deterministic}, improving upon the $O(\sqrt{n})$ update time achieved in \cite{frederickson1985data, eppstein1997sparsification}. It is a longstanding open question whether there exists a deterministic $polylog(n)$-time algorithm per operation in the worst case. Using the power of randomness, \cite{kapron2013dynamic} achieves worst case time $polylog(n)$ per operation. They propose a randomized algorithm which has worst case running time $O(\log^4 n)$ for edge insertion, $O(\log^5 n)$ for edge deletion, and $O(\log n /\log \log n)$ for each query. Moreover, the algorithm has one-sided error. That is, with probability less than $1/n^c$, the algorithm will declare that $u$ and $v$ are not connected even though they are. 
The reason is that in their algorithm (in fact, most of the algorithms in the literature), they use a tree to represent each maximal connected component. It is relatively easy to process each insert operation, because when one inserts an edge, if the two vertices are not previously connected, then we merge the trees which contain these two vertices. But it is hard to process delete operations, because once we delete a tree edge, we have to find a non-tree edge which reconnects two components. It might be possible that we do not find a reconnecting edge even though one exists.

\subsection{Previous Work}

The framework used by most previous works goes as follows. For given $G=(V,E)$, we maintain a spanning forest for $G$ via a dynamic tree structure supporting the following three kinds of operations efficiently.
\begin{itemize}
\item Link two disjoint trees to form a new tree. Specifically, link$(u,v)$ receives $u, v$ in two different trees of the forest $F$, then adds the edge $(u,v)$ to $F$;
\item Cut a tree edge so that the tree is split into two; 
\item Query whether two vertices are in the same tree. 
\end{itemize}

Actually, there are already many good ways to support all these three operations with deterministic worst case running time $O(\log n)$ per operation, since they are standard operations supported by dynamic trees such as the ET-tree \cite{henzinger1999randomized}. With a dynamic tree structure, we can deal with the operations in the dynamic connectivity problem as follows:

\begin{itemize}
\item When answering a query whether vertices $u$ and $v$ are connected, we just need to check if they are in the same tree in the dynamic tree structure; 
\item When inserting a new edge, if the edge is across two different trees, we link these two trees via this new edge in the dynamic tree structure so that we keep the invariant that the structure we are maintaining forms a spanning tree for the graph;
\item When deleting an edge, if the edge is not a tree edge, then we simply remove the edge without changing the dynamic tree structure. 
\end{itemize}

However, difficulty arises when we insert an edge which connects two vertices in the same tree, or when we delete a tree edge in the dynamic tree structure. Indeed, we can disregard new edges connecting vertices in the same tree, without affecting the invariant that we are maintaining a spanning forest. However, when we delete a tree edge, cutting the tree into two parts which we call $A$ and $B$, we need to check whether these two parts can reconnect, in order to maintain a spanning forest of $G$. More specifically, we want to find an edge across $A$ and $B$, i.e.\ an edge in the cutset $E\cap (A\times B)$ (or $E\cap (A\times (V\backslash A))$ because of the invariant) if one exists. Thus, if we follow this natural framework, we need to find a clever method to find an edge in the cutset. Also, we cannot afford to do nothing during insertions between edges already in the same tree, since such edges may become important after later deletions of tree edges.

One naive way is to undergo a brute force search of all possible edges. For example, iterate over all edges adjacent to $A$, and check whether there is one such edge with the other end in $B$. One breakthrough made in \cite{henzinger1996improved} is by observing that we can organize the edges in a clever way so that every edge is checked at most $O(\log n)$ times. Since every time an edge is checked one dynamic tree operation is performed, they get an amortized time $O(\log^2 n)$ per update. A further improvement was made to amortized time $O(\log n\log \log^3 n)$ in \cite{thorup2000near} via a more delicate data structure and randomized sampling technique. \cite{kapron2013dynamic} proposed a data structure to find such an edge in $polylog(n)$ worst case time with the guarantee that if there is at least one edge in the cutset, the data structure will find one with probability at least $1-1/n^c$. It is mainly based on the following two observations. 
\begin{enumerate}
\item First, if there is exactly one edge in the cutset $E\cap (A\times B)$, then the edge can be found in $O(\log n)$ worst time by a so-called XOR trick. it goes as follows: assign each edge a unique name (from $1$ to ${n \choose 2}$, for example), and for each vertex $v$, maintain the xor value of the names of edges adjacent to $v$, which we call $xor(v)$. Moreover, let $xor(A) = \bigoplus_{v\in A} xor(v)$ where $A$ is a set. One important property of the value is that for a vertex set $S\subseteq V$, the XOR of these xor values over all elements of $S$ equals the XOR of the edges in the cut set $E\cap (S\times (V\backslash S))$. Therefore, if in the dynamic tree structure we also maintain the XOR value of $xor(v)$ for all vertices $v$ in a tree, then we can get the name of the only edge in the cutset, which is $xor(A)$. The ET (Euler tour) tree introduced in \cite{henzinger1999randomized} can be used to maintain these XOR values (that is, we maintain a XOR value for each tree) in $O(\log n)$ time per operation in the worst case.
\item Second, if there is more than one edge in the cutset then we can use a sampling method as follows. For our graph $G=(V, E)$, we keep $2\log n+1$ subgraphs of $G$, namely $G_0,\ldots, G_{2\log n}$, where $G_i=(V,E_i)$ is subsampled from $G$ by selecting each edge in $E$ independently with probability $1/2^i$, i.e., $\forall e\in E, \Pr(e\in E_i)=1/2^i$.  Now instead of keeping track of the $xor$ values only on the graph $G=G_0$, we also maintain the $xor$ values for $G_i$, for $i=1,\ldots, 2\log n$. It is proven in \cite{kapron2013dynamic} that if $E\cap (A\times B)\ne \emptyset$ and $A,B$ are independent from $E_1,\ldots, E_{2\log n}$, with probability more than a positive constant, there is exactly one edge in $E_i\cap (A\times B)$ for some $i$. So with probability greater than a positive constant, we can find an edge in a non-empty cutset. If we duplicate $O(\log n)$ copies with independent randomness, we can boost the success possibility to $1-1/n^c$.  In order to know whether a level succeeded, it is sufficient to store all edges $E$ in a hash table, so we can check whether the XOR is actually a real edge by looking in the hash table in $O(1)$ time.

\end{enumerate}

Specifically,  Lemma~\ref{ds2} below summarizes the cutset data structure implemented in \cite{kapron2013dynamic}.

\begin{lemma}\label{ds2}[Cutset data structure]
There is a data structure maintaining a {\em dynamic forest} $F=\{T_i=\{V_{T_i}, E_{T_i}\}, i=1,\ldots\}$ for a dynamic graph $G=(V,E)$, where $\{V_{T_i}\}$ is a partition of $V$, and denote $E_F=\cup_{i}{E_{T_i}}$ as tree edges. It supports the following operations, with the requirement that when invoking function $insertTreeEdge$ or $deleteTreeEdge$, the input edge $e$ must be independent from the (random) outputs obtained by all previous calls to the function $outgoingEdge$.

\begin{itemize}
\item $insertTreeEdge(e=\{u,v\})$: insert edge $e$ into the forest $F$, combining the tree containing $u$ and the one containing $v$ together to be a single tree via $e$; 
\item $deleteTreeEdge(e=\{u,v\})$: remove edge $e$ from the tree containing it, splitting the tree into two new trees; 
\item $insertEdge(e)$: insert edge $e$ into $E$;
\item $deleteEdge(e)$: remove edge $e$ from $E$;
\item $tree(v)$: return the name of the tree containing vertex $v$ (each tree has a unique name);
\item $outgoingEdge(T)$: for input $T\in F$ ($T$ is specified by its name), let $C_T=E\cap (V_T\times (V\backslash V_T))$. (1) If $C_T=\emptyset$, return {\textbf{null}}; (2) If $C_T\ne \emptyset$, return an edge $e\in C_T$ with probability (which we call ``success probability'') at least $1/2$ or return {\textbf{null}} otherwise. 
\end{itemize}
$insertTreeEdge, deleteTreeEdge, insertEdge, deleteEdge$ and $outgoingEdge$ have worst case running time $O(\log^2 n)$ per invocation; $tree$ takes $O(\log n/\log \log n)$ time in the worst case per invocation.

Moreover, if we duplicate $O(\log n)$ copies of the cutset data structure with independent randomness, we can boost the success possibility to $1-1/n^c$. We call it ``boosted cutset data structure'' for future reference.
\end{lemma}

From now on, we will use the cutset data structure as a black box to construct dynamic connectivity algorithms. The cutset data structure cannot directly be applied to solve the dynamic connectivity problem for the following reason. Suppose we use it to maintain the spanning forest of the graph, after deleting a tree edge (cutting the tree into two parts which we call $A$ and $B$). Even though we can find an edge $e\in E\cap (A\times B)\ne \emptyset$ by invoking $outgoingEdge$, we cannot insert $e$ as a new tree edge into the cutset data structure because of the requirement in Lemma~\ref{ds2} that inserted tree edges be independent of all previous calls to $outgoingEdge$. \cite{kapron2013dynamic} circumvents this issue by implicitly maintaining the structure of a Boruvka tree.  A Boruvka tree records the process of how the vertices merge together to form larger and larger connected components. Starting from isolated nodes on layer $0$, i.e.\ single vertices of $G$, a Boruvka tree is built by coupling connected nodes\footnote{A node in a Boruvka tree corresponds to a set of vertices connected to each other. We say two nodes in a Boruvka tree are connected if there is an edge across the vertex sets corresponding these two nodes. If on the $i$-th layer one node is not matched with another node, then we just make an identical node on the $(i+1)$-st layer, so that on each layer the nodes give a partition of $V$.} on the $i$-th layer to form the nodes on the $(i+1)$-st layer, and we stop the process until we cannot merge any two nodes on some layer. 

Specifically, \cite{kapron2013dynamic} maintains $\ell+1$ forests $F_0,\ldots, F_{\ell}$, where $\ell=\log n$. We refer to $F_i$ as the forest at layer $i$ for $i=0,\ldots, \ell$. They use a boosted cutset data structure $cutset_i$ for each $F_i$. 
For each $i>0$, $F_i$ contains all the edges in $F_{i-1}$, and $cutset_{i-1}$ is used to provide additional edges to $F_{i}$ (that is, $E_{F_0}\subseteq \ldots \subseteq E_{F_{\ell}}$). For a tree $T$ on any layer $i$ where $0\le i <\ell$, we say it is {\em merged} (or {\em matched}) with a partner iff the tree containing $T$ on layer $i+1$ is strictly larger than $T$. \cite{kapron2013dynamic} maintains the invariant (with overwhelming probability $1-1/n^c$, where $c$ is a large constant) that any non-maximal tree on any layer $i$ is merged with a partner. That is to say,

\begin{invariant}\label{invariant1}
For any layer $i=0,\ldots, \ell-1$, for any tree $T\in F_i$, if $T$ is not a maximal tree in $G$, then there is a tree $T'\in F_{i+1}$ such that $V_{T}\subsetneq V_{T'}$. 
\end{invariant}

If Invariant~1 holds, then $F_{\ell}$ is a spanning forest for $G$, because every non-maximal tree on layer $i$ has size at least $2^i$, so that every tree on layer $\ell$ is maximal. To answer Query($u,v$), it is sufficient to check whether $u$ and $v$ are contained in the same tree in $F_{\ell}$ by outputting $cutset_{\ell}.tree(u)=cutset_{\ell}.tree(v)$, which takes time $O(\log n/\log \log n)$. To maintain Invariant~1, \cite{kapron2013dynamic} performs $O(\log n)$ operations (other than $tree$) of the boosted version of cutset data structure for each insertion, and $O(\log^2 n)$ operations (other than $tree$) of the boosted cutset data structure for each deletion. Because the boosted cutset data structure has worst case time $O(\log^3 n)$ per operation, the worst case running time is $O(\log^4 n)$ per insertion, and $O(\log^5 n)$ per deletion.

\subsection{Our Techniques}

We observe that it is possible to use the cutset data structure without boosting, at the cost of only a constant factor of blowup in $\ell$. That is, let $\ell=C\log n$ where $C$ is a large constant. We can prove that these $\ell$ layers are sufficient to ensure that $F_{\ell}$ is a spanning forest with high probability, by ``almost'' maintaining Invariant~\ref{inv2} in our new algorithm.

\begin{invariant}\label{inv2}
For any layer $i =0,\ldots, \ell-1$, conditioned on the randomness of layer $0, \ldots, i-1$ (that is, the randomness used in $cutset_0,\ldots, cutset_{i-1}$), the probability (over the randomness of $cutset_i$) that a non-maximal tree on layer $i$ is merged with a partner is at least $1/2$.
\end{invariant}

More precisely, because there is a subtle dependence issue prohibiting us from arguing correctness directly for our new algorithm, we adapt a coupling method and instead prove that there is another algorithm (whose output is exactly identical to our new algorithm with overwhelming probability) that maintains Invariant~2. Furthermore, we show that Invariant~2 is sufficient to prove that $F_{\ell}$ is a spanning forest with overwhelming probability.

\subsection{Recent and Independent Work}
Recently and independently, by reducing the number of redundant copies of cutset data structures, \cite{gibb2015dynamic} has improved the dynamic connectivity data structure in \cite{kapron2013dynamic}. The work \cite{gibb2015dynamic} improves the update time for the deletion operation from $O(\log^5 n)$ to $O(\log^4 n)$ while keeping the insertion time as $O(\log^4 n)$. Meanwhile, our work shaves one $\log n$ factor from both update operations, achieving $O(\log^4 n)$ time for deletion and $O(\log^3 n)$ for insertion.

\section{Improved Algorithms for Dynamic Graph Connectivity}\label{sec:improved}

In this section, we present our new algorithm with worst case time $O(\log^3 n)$ per insertion, and $O(\log^4 n)$ per deletion. Our algorithm maintains $\ell+1$ layers where $\ell=C\log n$. For each layer $i$ ($0\le i \le \ell$), we maintain a dynamic forest $F_i$ of the dynamic graph $G=(V,E)$ by using a cutset data structure without boosting, denoted as $cutset_i$.  At all times we have for each $i>0$, $F_i$ contains all the edges in $F_{i-1}$, so $E_{F_0}\subseteq \ldots \subseteq E_{F_{\ell}}$. The data structure $cutset_{i-1}$ is used to provide additional edges to $F_{i}$, for $1\le i \le \ell$. That is, when a tree edge is deleted from layer $i$ splitting some tree $T$ into $T_1, T_2$, we attempt to find new matches for $T_1$ and $T_2$ using $cutset_{i-1}$. Our goal is to make $F_{\ell}$ a spanning forest of $G$, so that queries can be answered efficiently. 

We implement Insert and Delete in Algorithm~1 for our dynamic graph $G=(V,E)$, where initially $E=\emptyset$. Update is an auxiliary function for Delete. For Query($u,v$), we simply output $cutset_{\ell}.tree(u)=cutset_{\ell}.tree(v)$. In our implementation, we also need to maintain the size of the trees on each layer in order to check whether a tree is merged with a partner. It can be supported by the ET tree \cite{henzinger1999randomized}, for which the running time is $O(\log n)$ in the worst case per operation. Initially any $F_i$ has no tree edges.

In Algorithm~1, when inserting an edge $e=\{u,v\}$, (1) we first insert $e$ into $cutset_i$ for every $i$, because now $e$ belongs to $G$; (2) If $u$ and $v$ were not connected in $F_\ell$ before inserting $e$, then we add $e$ as tree edges for every $F_i$. When deleting an edge $e=\{u,v\}$, (1) we first remove $e$ from $cutset_i$ for every $i$ (note that at this point it might be possible that $e$ is a tree edge, but at the end of the delete operation $e$ will not be in any $F_i$); (2) For every layer $i$, if $e$ is an edge in $F_i$, then remove it from $F_i$; (3) Start with $i=0$ and go up to $\ell-1$, we perform Update for both $u$ and $v$ on layer $i$. An update operation for $u$ on layer $i$ ensures that if the tree containing $u$ in $F_i$ (denoted by $T_u$) is non-maximal, then with probability at least $1/2$, $T_u$ is merged with a partner. In $Update(i,u)$, we first check whether $T_u$ is merged with a partner. If it is not, we do as follows: try to find an outgoing edge $e'$ for $T_u$ using $cutset_i$, if we find one, then insert $e'$ into $F_{i+1},\ldots, F_{\ell}$. Note that we might need to first remove an edge $e''$ from $F_{k},\ldots, F_{\ell}$ for some $k>i$ to ensure that no cycle is formed. As shown in \cite[Section 4.3]{kapron2013dynamic}, the cycle edge $e''$ can be found in $O(\log n)$ time in the worst case, by maintaining a dynamic data structure whose update time is bounded by $O(\log ^2 n)$ for each edge insertion or deletion to the dynamic graph $G$.

\vspace{2mm}

\textbf{Algorithm~1}
\begin{itemize}
\item Insert($e=\{u,v\}$):

~~~~\textbf{for} $i=0,\ldots, \ell$ 

~~~~~~~~$cutset_i.insertEdge(e)$

~~~~\textbf{if} $cutset_{\ell}.tree(u)\ne cutset_{\ell}.tree(v)$

~~~~~~~~\textbf{for} $i=0,\ldots, \ell$

~~~~~~~~~~~~$cutset_i.insertTreeEdge(e)$

\item Delete($e=\{u,v\}$):

~~~~\textbf{for} $i=0,\ldots, \ell$

~~~~~~~~$cutset_i.deleteEdge(e)$

~~~~\textbf{for} $i=0,\ldots, \ell$

~~~~~~~~\textbf{if} $e\in E_{F_i}$ \texttt{//membership can be tested by maintaining a Boolean vector}

~~~~~~~~~~~~$cutset_i.deleteTreeEdge(e)$

~~~~\textbf{for} $i=0,\ldots, \ell-1$

~~~~~~~~$Update(i, u)$

~~~~~~~~$Update(i, v)$

\item $Update(i, u)$:

~~~~\textbf{if} $size(cutset_i.tree(u)) < size(cutset_{i+1}.tree(u))$ \textbf{return}

~~~~$\{v,w\}=e'\leftarrow cutset_i.outgoingEdge(cutset_i.tree(u))$

~~~~\textbf{if} $e'=\textbf{null}$ \textbf{return}

~~~~\textbf{for} $k=i+1,\ldots, \ell$

~~~~~~~~\textbf{if} $cutset_k.tree(v)=cutset_k.tree(w)$

~~~~~~~~~~~~Find edge $e''$ connecting $cutset_{k-1}.tree(v)$ and $cutset_{k-1}.tree(w)$ in $E_{F_k}$

~~~~~~~~~~~~\textbf{for} $j=k, \ldots, \ell$

~~~~~~~~~~~~~~~~$cutset_j.deleteTreeEdge(e'')$

~~~~~~~~~~~~\textbf{break}

~~~~\textbf{for} $j=i+1,\ldots, \ell$

~~~~~~~~$cutset_i.insertTreeEdge(e')$ 

\end{itemize}

Because only $O(\log n)$ cutset data structure operations are performed in Insert, and $O(\log^2 n)$ cutset data structure operations are performed in Delete (including ones performed in $Update$), the time complexity for insertion is $O(\log^3 n)$; the time complexity for deletion is $O(\log^4 n)$.

In order to prove its correctness, we use a coupling argument. We design Algorithm~2 identical to Algorithm~1 (using exactly the same random bits) except for the implementation of Insert. The only change is that we replace ``$cutset_{\ell}.tree(u)\ne cutset_{\ell}.tree(v)$'' in Algorithm~1 to be ``$u$ and $v$ are not connected in $G$''. For that purpose, Algorithm~2 will record the whole graph $G$. We don't care about the efficiency of Algorithm~2 since it is only a tool used in the analysis. We first prove that Algorithm~2 maintains Invariant~2 in Lemma~\ref{lemma1}, and then in Theorem~1 we prove that Invariant~2 ensures that $F_{\ell}$ in Algorithm~2 is a spanning forest of $G$ with overwhelming probability. $F_{\ell}$ being a spanning forest implies the algorithm answers correctly. Finally in Lemma~\ref{identical}, we prove that with overwhelming probability the outputs of Algorithm~1 and Algorithm~2 are identical, and thus prove the correctness of Algorithm~1.

\begin{lemma}\label {lemma1}
Algorithm~2 maintains Invariant~2. 
\end{lemma}

\begin{proof}
We can prove the lemma by induction on the number of updates to the graph. Assume it holds for the first $t$ updates. For the inductive step, we consider $(t+1)$-th update:
\begin{enumerate}[{(1)}]
\item[\textbf{Case 1:}] (we insert an edge $\{u,v\}$ where $u$ and $v$ were already connected in $G$) In this case, we do not change any tree. So the probability remains the same that any non-maximal tree in any layer $i$ is merged with a partner. Moreover, this insertion does not change any maximal tree into being non-maximal.
\item[\textbf{Case 2:}] (we insert an edge $\{u, v\}$ where $u$ and $v$ were not previously connected in $G$) Let $T_u$ and $T_v$ be the trees containing $u$ and $v$ respectively on layer $i$ (here we consider every $i=0,\ldots, \ell-1$). If both $T_u$ and $T_v$ are maximal, then the combined tree (denoted by $T_{uv}$) is also maximal, which does not affect the invariant. If either $T_u$ or $T_v$ is not maximal (say $T_u$ is not maximal), then $T_{uv}$ is also non-maximal. But now the probability that $T_{uv}$ is merged with a partner is at least the probability that $T_u$ is merged with a partner after the first $t$ updates, making the probability no less than $1/2$ by induction.

 For non-maximal trees that do not contain $u$ or $v$, the probability related to the invariant is unchanged.
 \item[\textbf{Case 3:}] (we delete an edge $e=\{u,v\}$) We consider every layer $i$ ($i=0,\ldots,\ell-1$) separately. For non-maximal trees that do not contain $u$ or $v$, the probability related to the invariant is unchanged; for trees containing $u$ or $v$, we analyze the following two cases separately.
\begin{enumerate}
\item In the first case we consider, $e$ is a tree edge in layer $i$. That is, $T_{uv}$ is split into $T_u$ and $T_v$ after removing $e$. Then the deletion algorithm will try to find a partner for $T_u$ and $T_v$ separately. If $T_u$ is not maximal, then over the randomness on layer $i$ (that is, the randomness in $cutset_i$), the probability that it is merged with a partner is at least $1/2$ via invoking $Update(i,u)$ (and similarly for $T_v$).
\item In the second case, $e$ is not a tree edge in layer $i$. Then Algorithm~2 tries to merge $T_u$ (and also $T_v$) with a partner via invoking $Update(i,u)$. This ensures both $T_u$ and $T_v$ satisfy the invariant. 
\end{enumerate}

\end{enumerate}
\end{proof}

We assume that $m$ (where $m$ is the number of operations to the dynamic graph) is upper bounded by $n^{c_1}$ where $c_1$ is a constant. Let $c=10c_1+10$, and let $C=50c$ where $C$ is the constant appearing in $\ell =C\log n$.

\begin{theorem} \label{theorem}
If Algorithm~2 maintains Invariant~2, then with probability greater than $1-1/n^c$, $F_{\ell}$ is a spanning forest for $G$.
\end{theorem}

\begin{proof}
Let $X_{i}$ denote the number of non-maximal trees on layer $i$, for $0\le i \le \ell$. It is sufficient to prove that $\Pr[X_{\ell}>0]\le n^{-c}$. We have $X_0\le n$, and $X_i$ is non-increasing. Let $Y_i$ ($1\le i\le \ell$) denote the indictor random variable for the event that $X_{i-1}=0$ or $X_{i}/X_{i-1}\le 7/8$. Because of the invariant, $\mathbb{E}_{cutset_i}[X_{i}|X_0,\ldots, X_{i-1}]\le \frac{3}{4}X_{i-1}$. By Markov inequality, we have $\Pr(Y_i=1)\ge 1/7$ conditioned on $Y_1,\ldots, Y_{i-1}$. By coupling with $n$ independently distributed $\text{Bernoulli}(1/7)$ and then applying Chernoff bound, we have 

$$\Pr(X_{\ell}>0)\le \Pr(X_{\ell}\ge 1) \le \Pr(\sum_{i=1}^{\ell}{Y_i}\le 8\log n)\le e^{-\frac{C\log n}{21} (1-56/C)^2}<n^{-c}.$$

\end{proof}

\begin{lemma}\label{identical}
For any dynamic graph operations of length $m$, over the randomness used in the cutset data structure, the outputs of Algorithm~1 and Algorithm~2 are identical for all the queries with probability at least $1-n^{-c/2}$.  
\end{lemma}

\begin{proof}
As long as $F_{\ell}$ is a spanning forest, the cutset data structures in Algorithm~1 and Algorithm~2 are exactly the same. So the probability is upper bounded by $n^{-c}\cdot m<n^{-c/2}$ by taking a union bound and then applying Theorem~\ref{theorem}.
\end{proof}

\section*{Acknowledgments}

The author would like to express gratitude for the wonderful research support of and inspiring discussions with Professor Jelani Nelson. The author also would like to thank Professor Valerie King for pointing out an error in his original manuscript.


\newpage

\begin{thebibliography}{KRKPT15}

\bibitem[EGIN97]{eppstein1997sparsification}
David Eppstein, Zvi Galil, Giuseppe~F. Italiano, and Amnon Nissenzweig.
\newblock Sparsification -- a technique for speeding up dynamic graph
  algorithms.
\newblock {\em Journal of the ACM (JACM)}, 44(5):669--696, 1997.

\bibitem[Fre85]{frederickson1985data}
Greg~N. Frederickson.
\newblock Data structures for on-line updating of minimum spanning trees, with
  applications.
\newblock {\em SIAM Journal on Computing}, 14(4):781--798, 1985.

\bibitem[GKKT15]{gibb2015dynamic}
David Gibb, Bruce Kapron, Valerie King, and Nolan Thorn.
\newblock Dynamic graph connectivity with improved worst case update time and
  sublinear space.
\newblock {\em arXiv preprint arXiv:1509.06464}, 2015.

\bibitem[HK99]{henzinger1999randomized}
Monika~R. Henzinger and Valerie King.
\newblock Randomized fully dynamic graph algorithms with polylogarithmic time
  per operation.
\newblock {\em Journal of the ACM (JACM)}, 46(4):502--516, 1999.

\bibitem[HT96]{henzinger1996improved}
Monika~R. Henzinger and Mikkel Thorup.
\newblock {\em Improved sampling with applications to dynamic graph
  algorithms}.
\newblock Springer, 1996.

\bibitem[KKM13]{kapron2013dynamic}
Bruce~M. Kapron, Valerie King, and Ben Mountjoy.
\newblock Dynamic graph connectivity in polylogarithmic worst case time.
\newblock In {\em Proceedings of the Twenty-Fourth Annual ACM-SIAM Symposium on
  Discrete Algorithms}, pages 1131--1142, 2013.

\bibitem[KRKPT15]{kejlberg2015deterministic}
Casper Kejlberg-Rasmussen, Tsvi Kopelowitz, Seth Pettie, and Mikkel Thorup.
\newblock Deterministic worst case dynamic connectivity: Simpler and faster.
\newblock {\em arXiv preprint arXiv:1507.05944}, 2015.

\bibitem[Tho00]{thorup2000near}
Mikkel Thorup.
\newblock Near-optimal fully-dynamic graph connectivity.
\newblock In {\em Proceedings of the thirty-second annual ACM symposium on
  Theory of computing}, pages 343--350. ACM, 2000.
\end{thebibliography}
\end{document}